\documentclass[12pt,leqeq]{article}
\usepackage{amssymb,amsthm}
\usepackage{amsmath}
\usepackage{amscd}
\usepackage{xy}
\xyoption{all}
\usepackage[dvips]{graphicx}
\newtheorem{thm}{Theorem}%[section] (If you want theorem numbered
\newtheorem{prop}[thm]{Proposition} %--> \begin\end{theorem,lemma,...}

\newtheorem{defn}[thm]{Definition}

\date{}
\begin{document}
\setlength{\baselineskip}{16pt}
\title{Indirect  Influences, Links Ranking, and Deconstruction of Networks}
\author{Jorge Catumba,  \  Rafael D\'iaz, \  Ang\'elica Vargas}
\maketitle

\begin{abstract}
The PWP map was introduced  by the second author as a tool for ranking nodes in  networks. In this work we extend this
technique so that it can be used to rank links as well. Applying the Girvan-Newman algorithm a
ranking method on links induces a deconstruction  method for networks, therefore we obtain new methods for finding
clustering and core-periphery structures on networks.
\end{abstract}

\section{Introduction}

Three problems stand out for their centrality in the theory of complex networks, namely,
hierarchization, clustering, and core-periphery. In hierarchization the aim is to find a ranking on the nodes
of a network reflecting the importance of each node.  Several methods have been proposed for such rankings, among
them  degree centrality, eigenvalue centrality, closeness centrality,
betweenness centrality \cite{new}, Katz index \cite{k},  MICMAC of Godet \cite{godet},
PageRank of Google \cite{b2,  meyer}, Heat Kernel of Chung \cite{chung, chungyau}, and
Communicability of Estrada and Hatano \cite{estrada}. We  are going to use in this work the PWP method
\cite{diaz} which we review in Section \ref{pwpwwny}; for comparison with other methods see
\cite{diaz, da}, and for applications and extensions see \cite{cd, cd2, diazgomez, da}. Clustering
consists in finding a suitable partition on the nodes of a network such that
nodes within blocks are highly connected, and nodes in different blocks are weakly connected.
The reviews \cite{fo, f, s} offer a fairly comprehensive picture of the many methods that have been proposed to attack this problem.
Maximization of Newman's   modularity function  and its extensions \cite{b, di, girvannewman, n} is a particularly interesting
approach since it proposes a mathematical principle instead of an algorithm.
Roughly speaking the core-periphery finding problem \cite{r, v, z} consists in peeling a network as
if it were an onion,  discovering the rings out which it is built. The inner rings form the core
of the network, the outer rings form its periphery.

In this work we argue that, within a certain framework, the three problems have a common root:
a  hierarchization method induces both a clustering finding method and a core-periphery finding method. Indeed, we provide
three alternative methods for  reducing clustering and core-periphery finding to hierarchization:
the first one via the dual network; the second one via the barycentric division network;  the third one
regards a link as a bridge, thus its importance
is proportional to its functionality and to the importance of the lands it joints.

We  work with double weighted directed  networks, i.e. weights defined both on links and on nodes, formally
introduced in Section \ref{pwpwwny} where we
generalize the PWP method so that it can be applied to double weighted networks.
In Section \ref{dww} we recall how a  ranking on links
induces, following the  Girvan-Newman algorithm, a network deconstruction method.
In Section \ref{dww4} we introduce the dual construction for double weighted networks and use it
to rank links, obtaining the corresponding clustering and core-periphery finding methods. In Section \ref{dwwn5}
we introduce the barycentric division construction for double weighted networks and  consider the corresponding
clustering and core-periphery finding methods. In Section \ref{dwwn6} we introduce the bridge approach to link ranking and
its corresponding clustering and core-periphery finding methods.
In Section \ref{s7} we illustrate the notions introduced along the paper by applying them
to a highly symmetric intuitively graspable network,  and also to
a more sophisticated network.

\section{PWP on Double Weighted Networks}\label{pwpwwny}

\begin{figure}[t]
\centering
\begin{tabular}{@{}cc@{}}
\includegraphics[width=6cm, height=3cm]{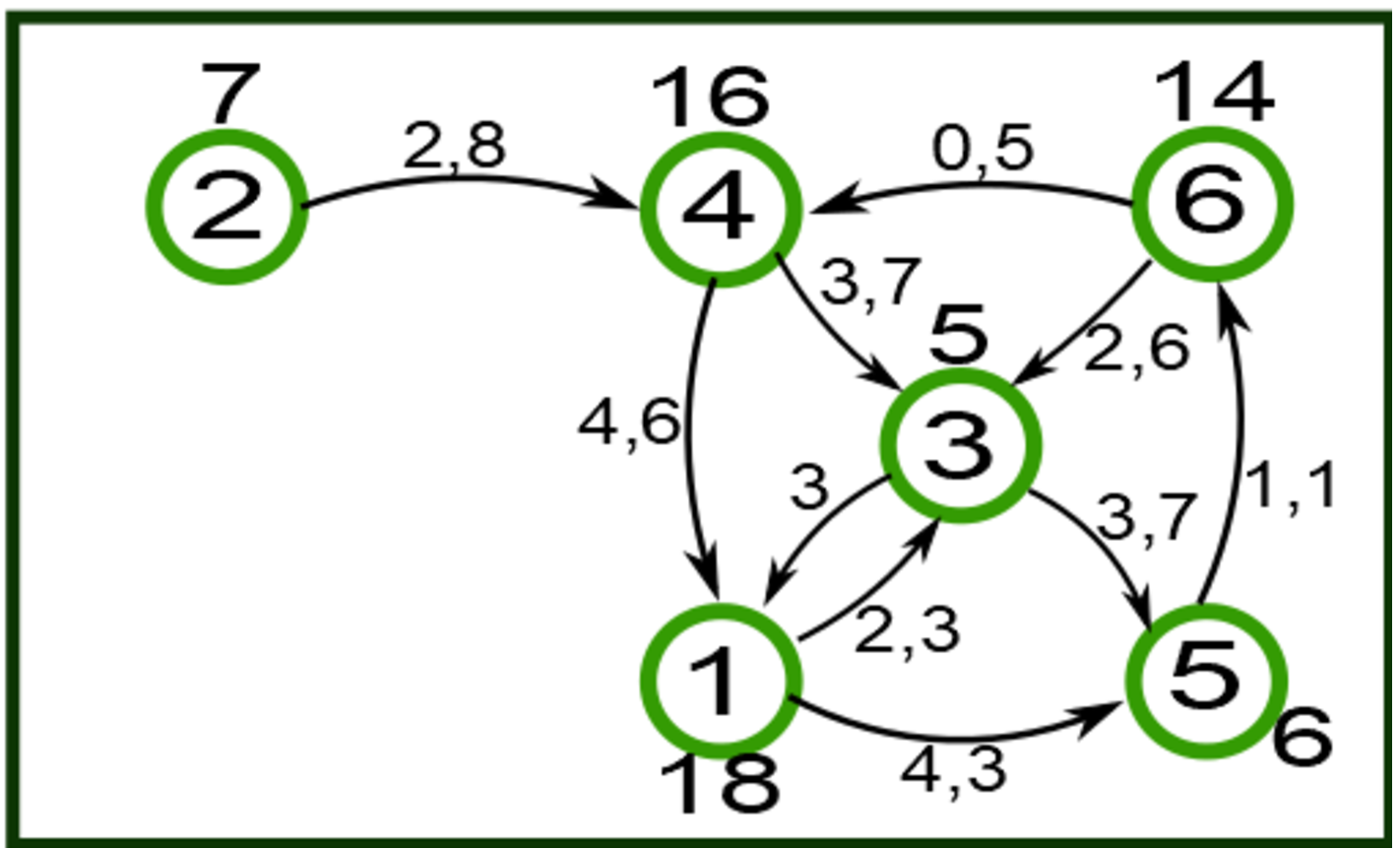}
\
\
\
\
\
\
\
\
\
\
\
\
\includegraphics[width=6cm, height=3cm]{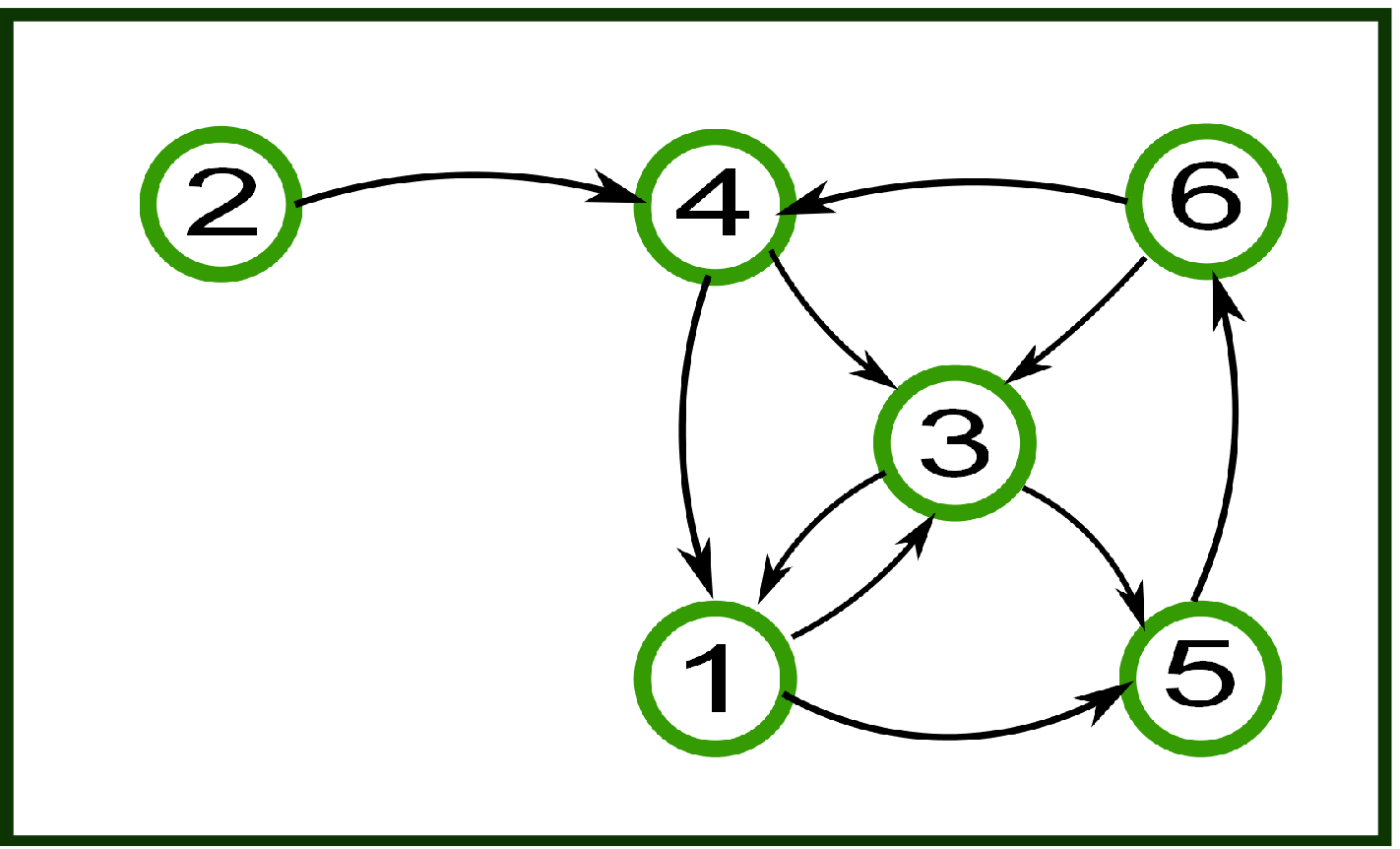}
\end{tabular}
\caption{Left: Weighted network $W.\  $ Right: Weighted network $Y.$}
\label{f1}
\end{figure}

Let $\mathrm{digraph} $ be the category of directed networks, $\mathrm{wdigraph} $ the category of directed networks with
weighted links, and  $\mathrm{wwdigraph} $ the category of directed networks with weighted nodes and weighted links,
i.e. objects  in  $\mathrm{wwdigraph} $ are tuples $(V, E,s,t, w, f) $ consisting of:
$-$A directed network $(V, E,s,t) $  with set of nodes $ V, $  set of links
$ E, $ and $ (s,t):E \rightarrow V\times V $ the source-target map. $-$A  map $ f: V \rightarrow \mathbb{R} $ giving
weight to nodes.
$-$A map $w:E \rightarrow \mathbb{R}$ giving weight to links. Figure \ref{f1} shows on the left the
 double weighted network $ W,  $
and on the right the double weighted network $ Y  $ with the same underlying network  and weights set to $ 1.$
 A morphism $ (\alpha, \beta):(V_1, E_1,s_1,t_1, w_1, f_1)
\rightarrow   (V_2, E_2,s_2,t_2,w_2, f_2)  $
in $ \mathrm{wwdigraph}$  is given by a pair of maps  $ \alpha: V_1 \rightarrow  V_2 $ and
 $ \beta: E_1 \rightarrow E_2 $ such that
$$(s_2,t_2)\circ \beta  = (\alpha \times \alpha)\circ (s_1,t_1), \ \ \ \ \ f_2(v_2)  =
 \sum_{v_1 \in V_1,\ \alpha v_1=v_2}f_1(v_1), \ \ \ \ \  w_2(e_2)  =  \sum_{e_1 \in E_1,\ \beta e_1=e_2}w_1(e_1).$$
Figure \ref{mor} displays a morphism in  $\mathrm{wwdigraph} ,$ represented by thick arrows, with $W$ as domain.
To each double weighted network on $[n]  =  \{1,...,n \} $ we associate a matrix-vector pair
$(D,f) \in \mathrm{M}_n(\mathbb{R}) \times \mathbb{R}^n  $  consisting of the adjacency matrix $D$ and the
vector $f$ of weights on nodes:
$$D_{ij} =  \sum_{e\in E,  \ t(e)=i,  \ s(e)=j}w(e) \ \ \ \ \  \mbox{and} \ \ \ \ \ f_j =  \mbox{weigth of} \ j.$$

\begin{figure}[t]
\centering
\includegraphics[width=14cm, height=5cm]{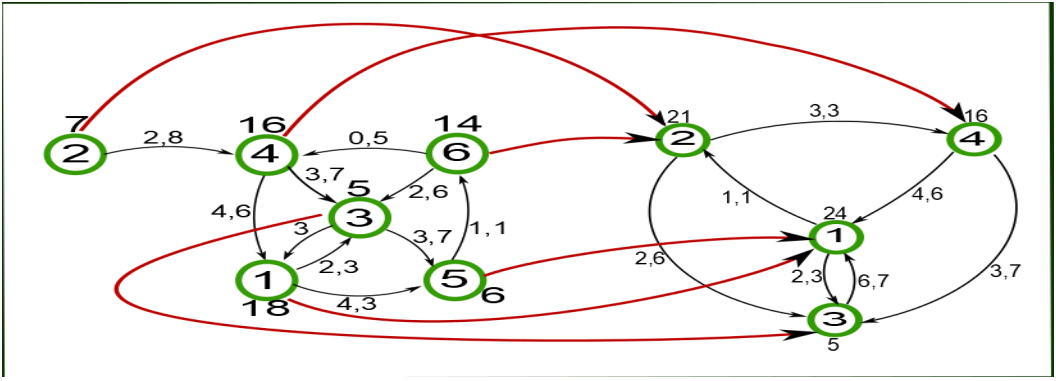}
\caption{Morphism between double weighted directed graphs.}
\label{mor}
\end{figure}
A weighted network   without multiple links on $[n] $ and the pair $(D,f)  $ encode, essentially,
the same information. For simplicity we usually work with networks without multiples links.
Morphisms between weighted networks without multiples links are defined for matrix-vector pairs,
say from  $ (D,f) \in \mathrm{M}_n(\mathbb{R}) \times \mathbb{R}^n  \  $ to
  $\ (E,g) \in \mathrm{M}_m(\mathbb{R}) \times \mathbb{R}^m  $ by a map $\alpha:[n] \rightarrow [m] $ such that
$E_{ij} = \sum_{\alpha(k)=i, \alpha(l)=j}D_{kl} $ and  $g_j  =  \sum_{\alpha(l)=j}f_l.$
We have maps $ \mathrm{digraph}    \rightarrow  \mathrm{wdigraph}   \rightarrow   \mathrm{wwdigraph}  $
where the first map gives weight $1 $ to links, and the second map keeps the weight on links and gives
weight $ 1 $ to nodes. The product
$ (V_1 \times V_2, E_1\times E_2 ,(s_1,s_2),(t_1,t_2), w_1\times w_2, f_1\times f_2)$ of double
weighted networks  $(V_1, E_1,s_1,t_1, w_1, f_1) $ and   $  (V_2, E_2,s_2,t_2,w_2, f_2) $
is such that $w_1\times w_2(e_1,e_2) =   w_1(e_1)w_2(e_2)  $  and
$  f_1\times f_2(v_1,v_2)  =   f_1(v_1)f_2(v_2). $ Similarly,  disjoint union of
double weighted networks can be defined.

  \begin{figure}[t]
\centering
\begin{tabular}{@{}cc@{}}
\includegraphics[width=7cm, height=4cm]{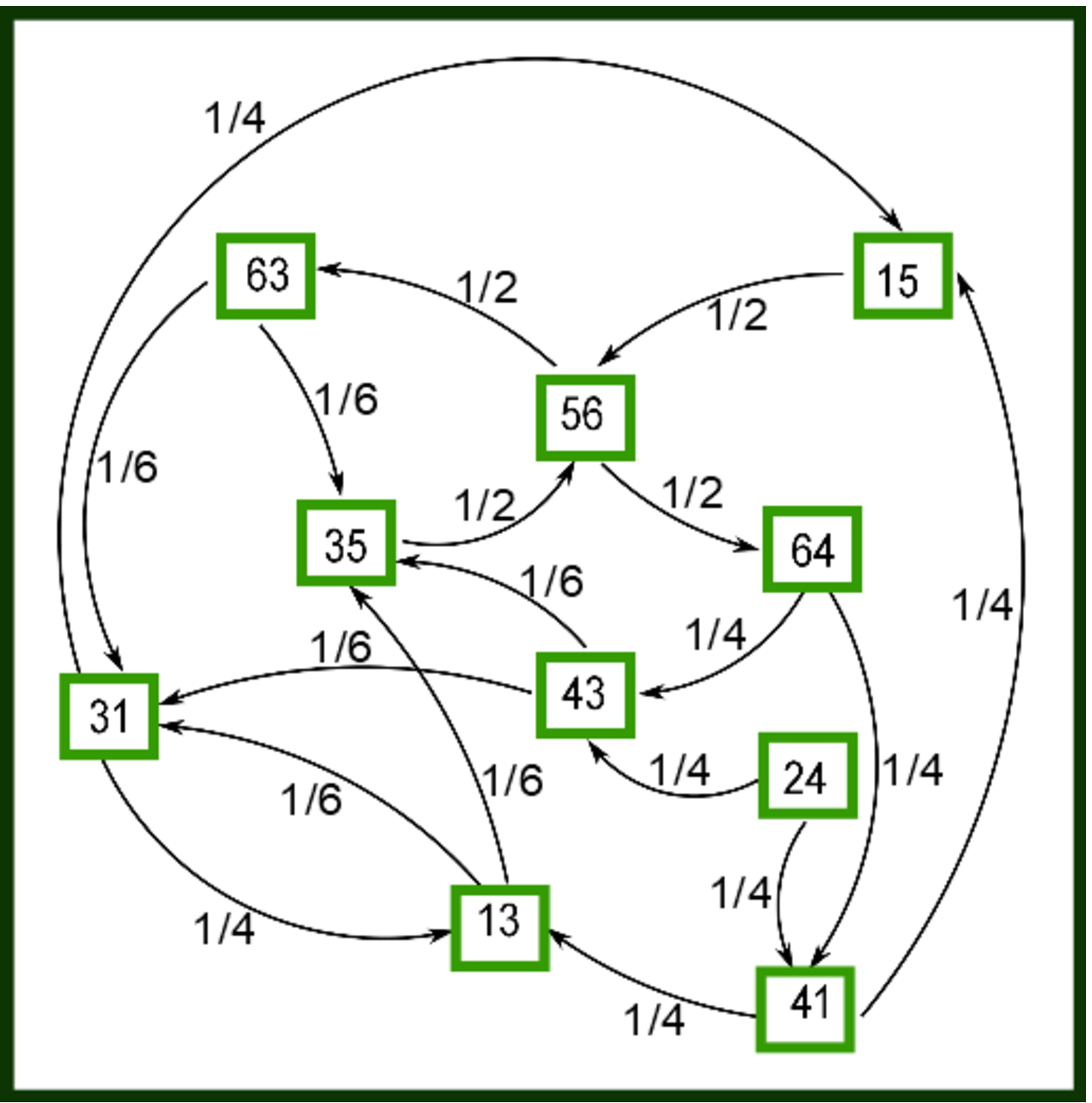}
\
\
\
\
\
\
\
\
\
\includegraphics[width=7cm, height=4cm]{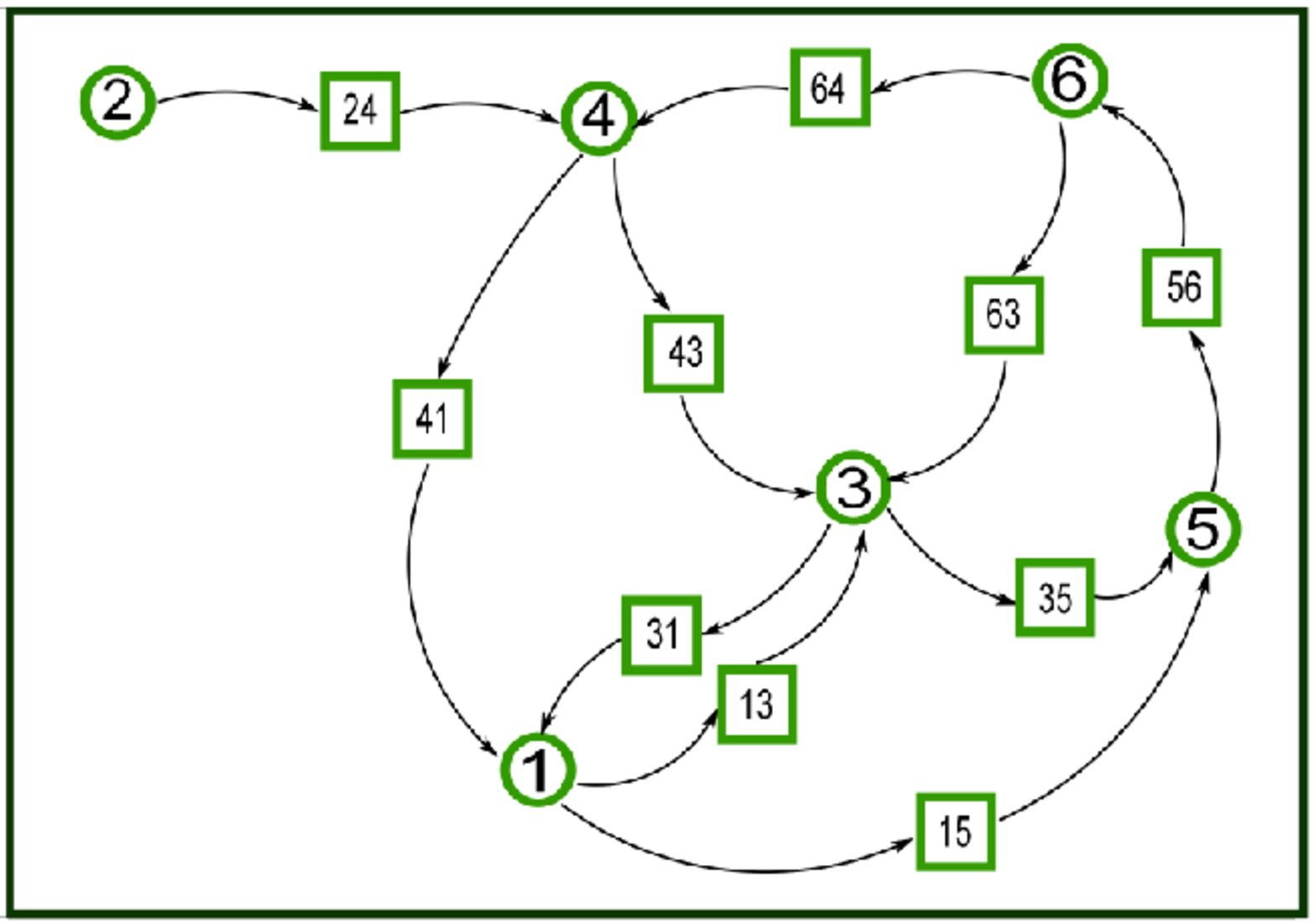}
\end{tabular}
\caption{Left: Dual Network $Y^{\star}$. \  Right: Barycentric  Network $Y^{\circ}$.}
\label{jjj}
\end{figure}

The PWP  map depends on a parameter $\lambda \in \mathbb{R}_{\geq 0}$ and is useful for measuring
 \textit{indirect} influences on networks. Assume as  given a weighted directed network
 with associated matrix $ D \in \mathrm{M}_n(\mathbb{R})$ measuring the \textit{direct}
 influence that each node exerts on the other nodes. The PWP map
$T:\mathrm{M}_{n}(\mathbb{R}) \rightarrow  \mathrm{M}_{n}(\mathbb{R}) $
sends  $ D  $  to the matrix of indirect influences  $T=T(D, \lambda)$ given by:
$$T(D, \lambda)\  = \    \frac{e^{\lambda D}  -  I}{e^{\lambda}  -  1}
\  =  \  \frac{\sum_{k=1}^{\infty}D^k\frac{\lambda^k}{k!}}{\sum_{k=1}^{\infty}\frac{\lambda^k}{k!}}  \ =
\ \sum_{k=1}^{\infty}\frac{\lambda^k}{e^{\lambda}  -  1}\frac{D^k}{k!}
\ = $$
$$\sum_{k=1}^{\infty}\bigg(\sum_{j=0}^{\infty}B_j\frac{\lambda^{j+k-1}}{j!} \bigg)\frac{D^k}{k!} \ =
\ \sum_{k=0}^{\infty}\bigg(\sum_{l=0}^{k}(k)_lB_{k-l}\frac{D^{l+1}}{(l+1)!}\bigg)\frac{\lambda^k}{k!} \ = $$
$$D \ + \ \frac{D^2}{2}\lambda\ + \ \bigg(\frac{D}{6}+
\frac{D^3}{3}\bigg) \frac{\lambda^2}{2!} \ + \ \bigg(\frac{D^2}{4} + \frac{D^4}{4}\bigg) \frac{\lambda^3}{3!}\
+ \ \bigg(-\frac{D}{30}+\frac{D^3}{3} + \frac{D^5}{5}\bigg)\frac{\lambda^4}{4!}\ + \cdots $$ where
$ (k)_l=  \frac{k!}{(k-l)!}$ and  $ B_j \in \mathbb{Q}$ are the Bernoulli numbers. From the above expression
we see that $ T(D, \lambda) $ is a one-parameter deformation of the adjacency matrix $ D  $
in the sense that  $  T(D,0)=D; $ replacing $ D $ by $ T(D, \lambda) $ one obtains
a one-parameter deformation of all network concepts defined in terms of the matrix $D$
of direct influences.  As an example we introduce  an one-parameter deformation
of the Girvan-Newman modularity function that takes indirect influences into account.
The Girvan-Newman modularity function $ Q: \mathrm{Par}[n] \rightarrow \mathbb{R}, $ defined on the set of
partitions on the nodes of a (non-negative) weighted directed network $ D,  $ is given  by
$$Q(\pi) = \sum_{ i\sim j}\bigg( \frac{D_{ij}}{m} -
\frac{d_i^{\mathrm{in}}}{m}\frac{d_j^{\mathrm{out}}}{m}\bigg)$$
where the sum is over pair of nodes in the same block of $\pi$, and
$$m=\sum_{i,j\in [n]}D_{ij}>0, \ \ \ \ \ d_i^{\mathrm{in}}=\sum_{j\in [n]}D_{ij}, \ \ \ \ \
d_j^{\mathrm{out}}=\sum_{i\in [n]}D_{ij}.$$
Turning on indirect influences we obtain the one-parameter deformation of the modularity function
$ Q_{\lambda}: \mathrm{Par}[n] \rightarrow \mathbb{R} \ $ $(\lambda \in \mathbb{R}_{\geq 0}, \ Q_0=Q) \ $ given by
$$Q_{\lambda}(\pi) = \sum_{ i\sim j}\bigg( \frac{T_{ij}(\lambda)}{M(\lambda)} -
\frac{E_i(\lambda)}{M(\lambda)}\frac{F_j(\lambda)}{M(\lambda)}\bigg) \ \ \  \mbox{where} \ \ \
M(\lambda)=\sum_{i,j\in [n]}T_{ij}(\lambda)>0.$$
In several examples, including the networks study in the closing sections and quite a few randomly generated networks, $ Q_{\lambda}$ is a monotonically
decreasing function of $\lambda$, a result intuitively appealing since turning on indirect influences
makes networks "more connected."

We use  the matrix  of indirect influences $T$ to impose rankings on nodes of  networks.
Let $\mathrm{ranking}(X)$ be the set of rankings on $X,$ i.e. a pre-order $ \leq $
on $X$ for which there is a map
$ f: X \rightarrow \mathbb{N}$ such that $i \leq j$ if and only if $f(i)\leq f(j).$
Equivalently, a ranking  is given by a partition on $ X  $ together with a linear ordering
on the blocks of the partition, thus the exponential generating series for rankings is
$$ \sum_{n=1}^{\infty}\big| \mathrm{ranking}[n]  \big|\frac{x^n}{n!}  =   \frac{x}{1-x} \circ (e^x -1) =
\frac{e^x-1}{2-e^x}. $$ Rankings on $[n]$ are isomorphic if there is a bijection $\alpha:[n] \rightarrow [n] $
that preserves pre-orders. The number of isomorphism classes of rankings on  $ [n] $ is equal to the number of
compositions on  $n $ so, see \cite{andrews}, there are $2^{n-1} $ non-isomorphic rankings.

In our previous works \cite{diaz, diazgomez, da} we  used the maps
$ \mathrm{E},  \mathrm{F}, \mathrm{I}: \mathrm{M}_n(\mathbb{R})   \rightarrow   \mathrm{ranking}[n] $
where the rank of a node in the respective pre-orders is given, setting $T=T(D, \lambda),$ by
$$\mathrm{E}_i =  \sum_{j=1}^n T_{ij}, \ \ \ \ \ \  \mathrm{F}_i  =  \sum_{j=1}^n T_{ji}, \ \ \ \ \  \mbox{and} \ \ \ \ \
  \mathrm{I}_i  =  \sum_{j=1}^n (T_{ij}  +  T_{ji})  .$$

\noindent We call these rankings the ranking by indirect dependence, indirect influence, and importance.
  The ranking  by importance on networks $ W $ and $Y$  from Figure \ref{f1} are:
\begin{center}
\begin{tabular}{|c|c|}
  \hline
 \  \   & Nodes Ranking by Importance \\  \hline
  \   Network Y \ & \ \ $3 >  1  >  5  >  6  >  4  >  2$ \ \ \\  \hline
  \ Network W \ & \ \ $5 >  3  > 1  >  4  >  6  >  2$ \ \ \\  \hline
 \end{tabular}
\end{center}

\begin{figure}[t]
\centering
\begin{tabular}{@{}cc@{}}
\includegraphics[width=7cm, height=4cm]{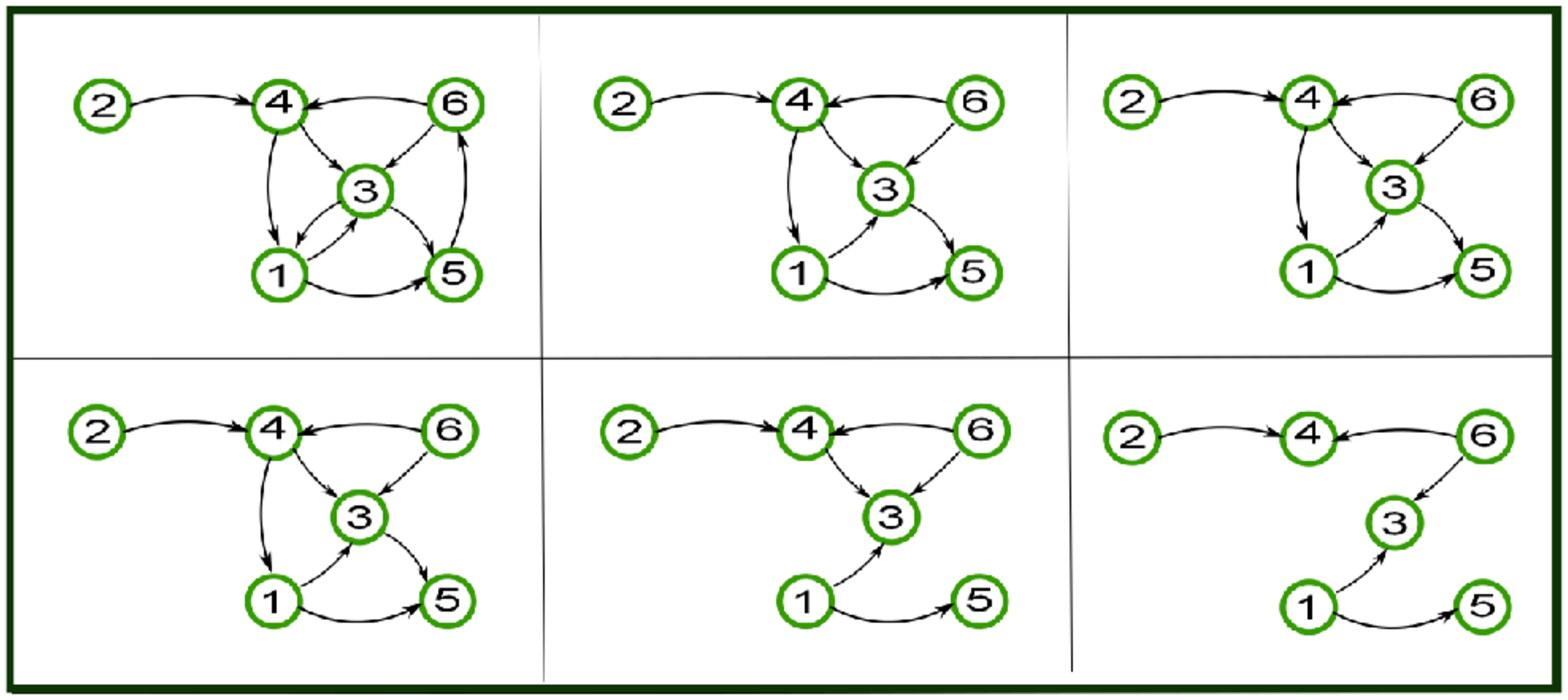}
\
\
\
\
\
\
\
\includegraphics[width=7cm, height=4cm]{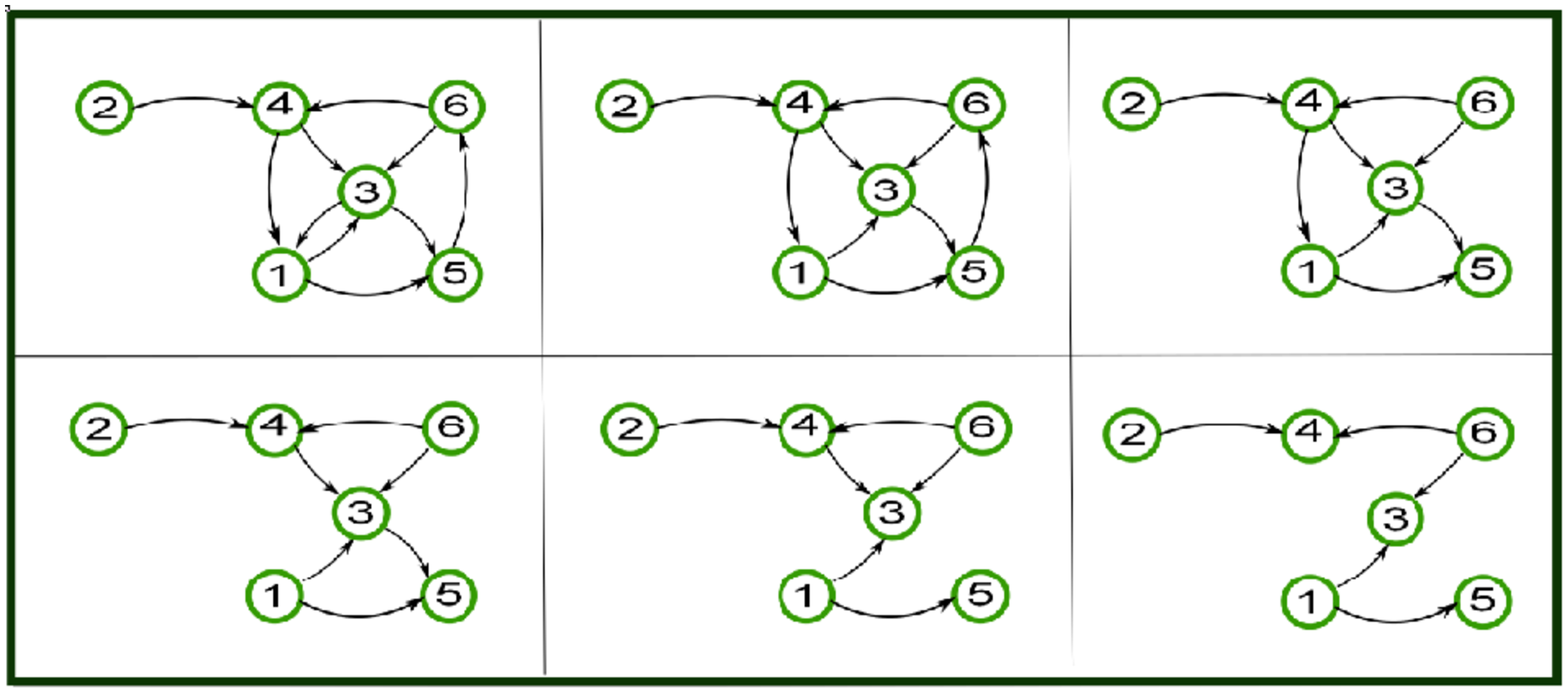}
\end{tabular}
\caption{Left: Clustering by importance process using dual or bridge constructions. Right: Clustering by  importance process using barycentric construction.}
\label{cdb}
\end{figure}
\noindent We extend the PWP map to double weighted networks  by defining a map
$$ T:\mathrm{M}_{n}(\mathbb{R})\times \mathbb{R}^{n} \rightarrow  \mathrm{M}_{n}(\mathbb{R})   $$
sending  a pair $(D,f) $ measuring direct influences and weight of nodes,  to the matrix
$T =  T(D,f,\lambda) $ measuring indirect influences among nodes. Let
$_{\bullet}: \mathrm{M}_{n}(\mathbb{R})\otimes \mathbb{R}^{n}  \rightarrow  \mathrm{M}_{n}(\mathbb{R}) $
be the linear map given on $D \otimes f$  by
$  (D_{\bullet} f)_{ij}  =   D_{ij}f_j. $

\begin{defn}{\em
 The PWP map $ T:\mathrm{M}_{n}(\mathbb{R})\times \mathbb{R}^{n}  \rightarrow  \mathrm{M}_{n}(\mathbb{R})  $
is given for $\lambda \in \mathbb{R}_{\geq 0} $ by
$$T(D,f, \lambda)  =    \frac{e^{\lambda D_{\bullet}f}  -  I}{e^{\lambda}  -  1}  =
  \frac{\sum_{k=1}^{\infty}(D_{\bullet}f)^k\frac{\lambda^k}{k!}}{\sum_{k=1}^{\infty}\frac{\lambda^k}{k!}}  =
  \sum_{k=0}^{\infty}\bigg(\sum_{l=0}^{k}\frac{(k)_lB_{k-l}}{(l+1)!}
  (D_{\bullet}f)^{l+1}\bigg)\frac{\lambda^k}{k!} . $$}
\end{defn}

Next we give an explicit formula for the entries of the PWP matrix $T $ of indirect influences
for double weighted networks, which implies the probabilistic interpretation for $T$ given below.

\begin{prop}\label{v}
{\em Let $(V, E,s,t, w, f)$ be a double weighted network, the indirect influence of node $ v  $  on node $  u $
according to the PWP map  is given by
$$T_{uv}(\lambda) =  \frac{1}{e^{\lambda}-1} \sum_{k=1}^{\infty} \sum_{\underset{te_k=u,\ te_i=se_{i+1},
 \ se_1=v }{e_k,...,e_1}} w(e_k)f(se_k)....w(e_1)f(se_1)\frac{\lambda^k}{k!}.$$
Equivalently, the matrix of indirect influences $ T(D,f, \lambda) $ is given  by
$$T_{ij}  =
\frac{1}{e^{\lambda}-1}\sum_{k=1}^{\infty}\Big(\underset{i=i_k, \ldots , i_0=j}{\sum}
 D_{i_ki_{k-1}}f_{i_{k-1}} \cdots D_{i_1i_0}f_{i_0} \Big)\frac{\lambda^k}{k!} .$$}
\end{prop}

\begin{figure}[t]
\centering
\begin{tabular}{@{}cc@{}}

\includegraphics[width=11cm, height=4cm]{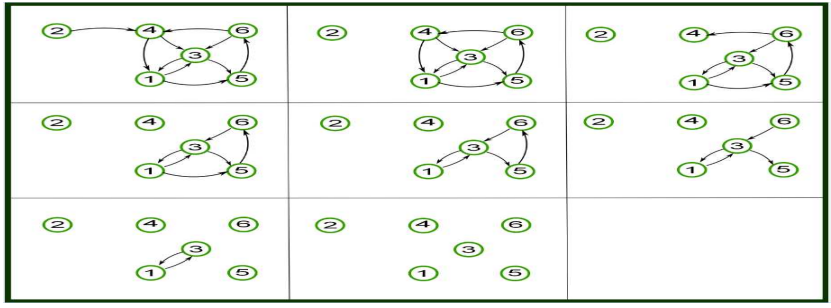}
\end{tabular}
\caption{Core-periphery finding  process by importance on network $Y$ for both the dual and bridge constructions.}
\label{cpdb}
\end{figure}

We regard $\frac{\lambda^k}{(e^{\lambda}  -  1)k!} $ as a probability measure on $ \mathbb{N}_{>0}. \ $
Note that  $ \frac{\lambda^k}{k!} \leq  \frac{\lambda^{k+1}}{(k+1)!}  $
if and only if $  k \leq \lambda -1,$ therefore $   \frac{\lambda^k}{(e^{\lambda}-1)k!} $
achieves its maximum at $ \lfloor \lambda \rfloor.  $ The mean and variance of
$ \frac{\lambda^k}{(e^{\lambda}-1)k!} $ are respectively  $ \frac{\lambda e^{\lambda}}{e^{\lambda}-1} $
 and $ \frac{\lambda e^{2\lambda} - \lambda e^{\lambda} - \lambda^2e^{\lambda}}{(e^{\lambda}-1)^2}.  $ By Chebyschev's
theorem we have for $ a>0  $ and $l \in \mathbb{N}_{>0}  $ that
$$\mathrm{prob} \bigg( \big|l- \frac{\lambda e^{\lambda}}{e^{\lambda}-1}\big| \geq a \frac{\lambda e^{2\lambda} - \lambda e^{\lambda} - \lambda^2e^{\lambda} }{(e^{\lambda}-1)^2} \bigg) \leq  \frac{1}{a^2}.$$ For example setting
$ \lambda=1, a=10 $ we get that $ \mathrm{prob} \big( l \geq 9 \big) \leq  10^{-2}.  $

Proposition \ref{v} implies a probabilistic interpretation for  $ T_{ij}  $  under the
assumption that $ D_{ij}, f_i \in [0,1]  $  give the probabilities that the
link $j \rightarrow i $ and the node $ i  $ be active, respectively.
Under these assumptions it is natural to let the probability
that a path $ (i_0,i_1,...,i_k)  $ be active be
$\ D_{i_ki_{k-1}}f_{i_{k-1}} \cdots D_{i_1i_0}f_{i_0} \frac{\lambda^k}{(e^{\lambda}-1)k!} ,\ $
i.e. being active is an independent property among the components
of a path, links and nodes; a length dependent correcting factor
$  \frac{\lambda^k}{(e^{\lambda}-1)k!}  $ is included making long path less likely to be active.

\begin{figure}[t]
\centering
\begin{tabular}{@{}cc@{}cc@{}}
\includegraphics[width=5.1cm, height=6.2cm]{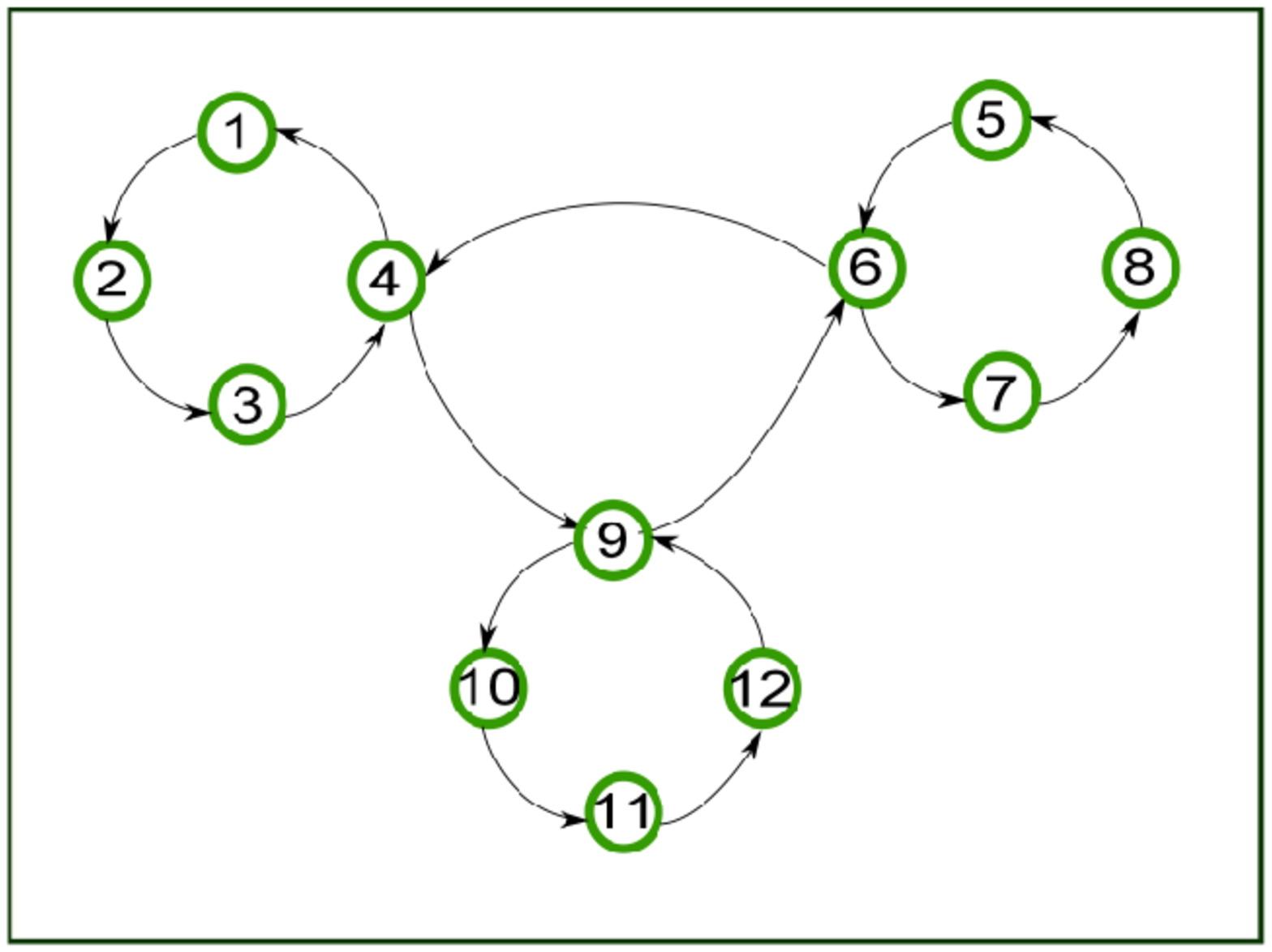}
\

\includegraphics[width=5.1cm, height=6.2cm]{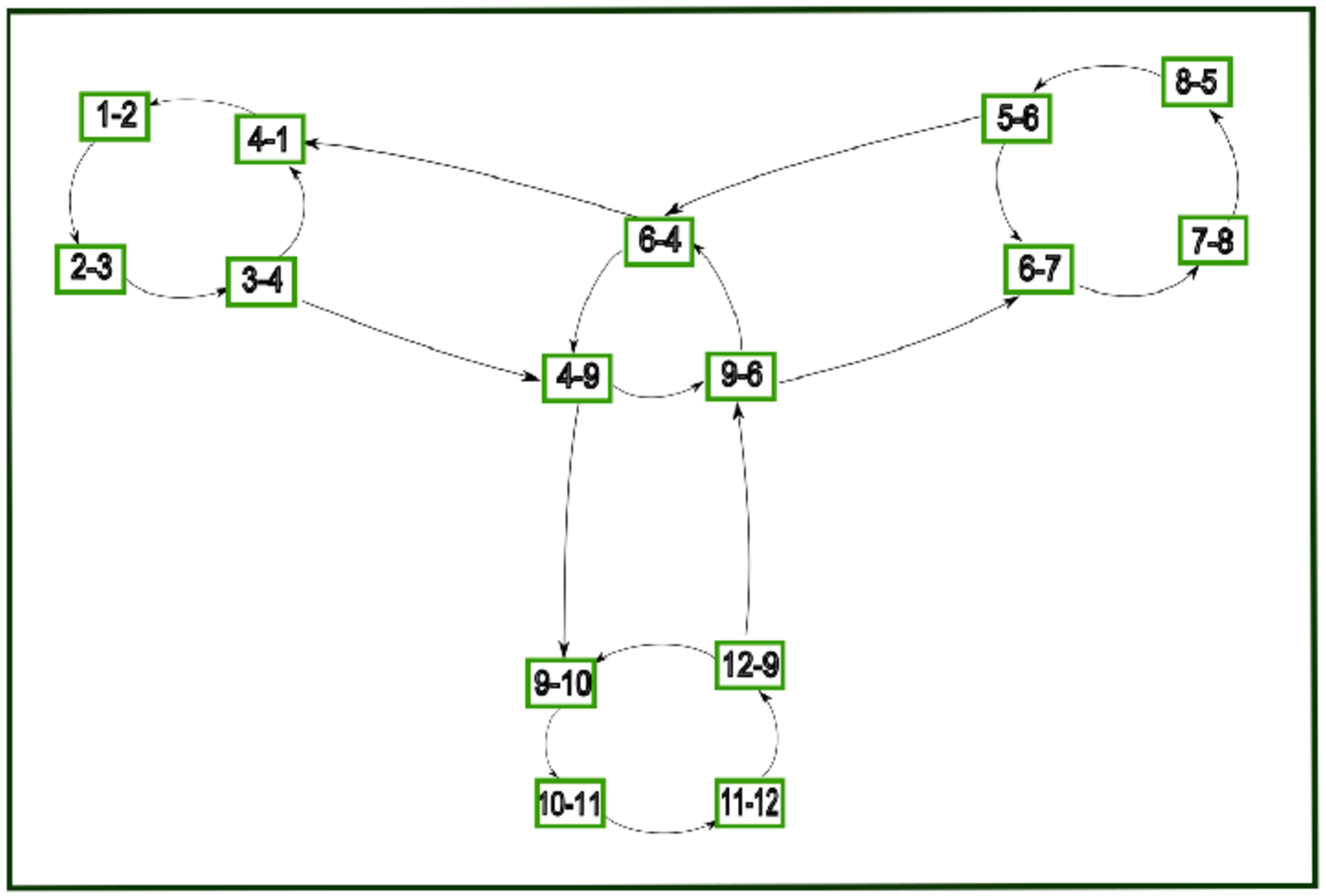}
\

\includegraphics[width=5.1cm, height=6.2cm]{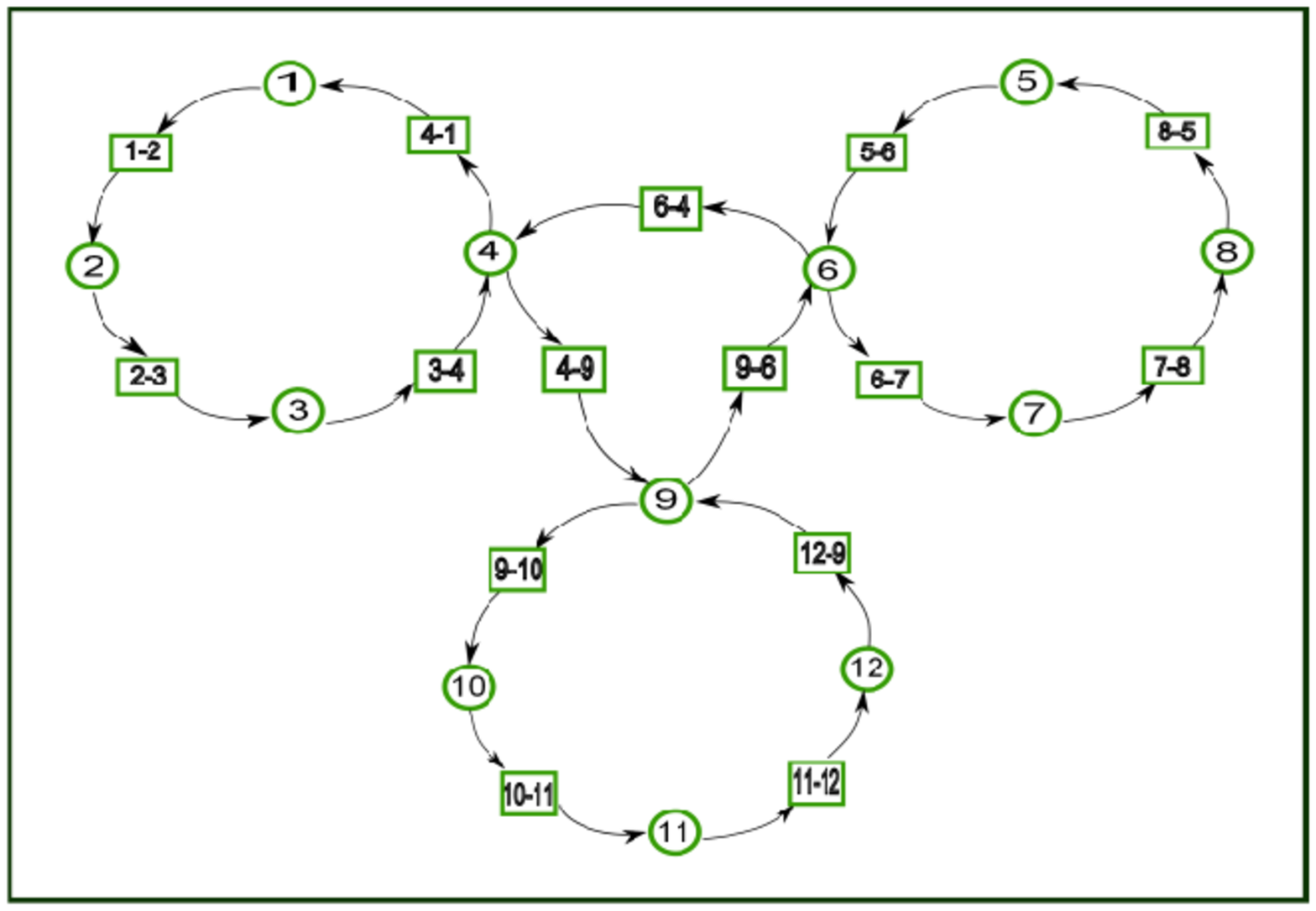}
\end{tabular}
\caption{Left: Network $S$. Centre: Barycentric Network  $S^{\circ}$.
Right: Dual Network $S^{\star}.$}
\label{jjj1}
\end{figure}

\begin{thm}
{\em  Consider a double weighted network with $ D_{ij},  f_i \in [0,1] $ giving the probabilities that the
link $j \rightarrow i $ and the node $ i  $ be active, respectively.
The indirect influence $  T_{ij}  $ of node $ j  $ on node $ i  $  is the expected number of active paths  from $  j  $ to $ i. $
}
\end{thm}

\begin{proof}Let $ \Omega $ be a probability space provided with
independent random variables $ \hat{D}_{ij},  \hat{f}_i : \Omega \rightarrow \{0,1 \} $ such that
 $ E(\hat{D}_{ij})= D_{ij} $
and $ E( \hat{f}_i)=f_i. $ Let  $ (\Omega \times\mathbb{N}_{n \geq 1}, p) $
be the probability space with $ p(w,k)=p(w)\frac{\lambda^k}{(e^{\lambda}-1)k!}, $
and consider the random variables
$\hat{T}_{ij}: \Omega \times \mathbb{N}_{\geq 1} \rightarrow [0,1]  $ given by
$$ \hat{T}_{ij}(w,k)=
\underset{i=i_k, \ldots , i_0=j}{\sum}  D_{i_ki_{k-1}}(\omega)f_{i_{k-1}}(\omega)
 \cdots D_{i_1i_0}(\omega)f_{i_0}(\omega). $$
The expected number of active paths $E\hat{T}_{ij}$  from $ j  $ to $ i $ is given by
$$\sum_{k=1}^{\infty}\underset{i=i_k, \ldots , i_0=j}{\sum} E
\hat{D}_{i_ki_{k-1}} E\hat{f}_{i_{k-1}}
 \cdots E\hat{D}_{i_1i_0}E\hat{f}_{i_0}\frac{\lambda^k}{(e^{\lambda}-1)k!}  =$$
$$ \sum_{k=1}^{\infty}D_{i_ki_{k-1}}f_{i_{k-1}} \cdots D_{i_1i_0}f_{i_0}
 \frac{\lambda^k}{(e^{\lambda}-1)k!} = T_{ij}.$$
\end{proof}

\section{Network Deconstruction from Links Ranking}\label{dww}

Applying the Girvan-Newman \cite{girvannewman} clustering algorithm
we describe how to deconstruct a network assuming as given  a method for ranking links on networks.
Let the connected components of a directed network be the equivalence classes of nodes under the
equivalence relation generated by  adjacency.  The Girvan-Newman
algorithm iterates the following procedures: -Compute the ranking of links. -Remove the  links of highest rank.
The algorithm stops when there are no further links, and outputs a forest of rooted trees, dendrogram, determined by the following
properties: -The roots are the connected components of the original network. -The leaves are the nodes of
the original network. -The internal nodes are the connected components of the various networks
that arise as the procedures above are iterated. -There is an arrow from node
$a $ to node $ b$ if and only if $ a \subset b $ and there is no node $ c  $ such that  $a \subset c \subset b .$

\begin{figure}[t]
\centering
\includegraphics[width=16cm, height=7cm]{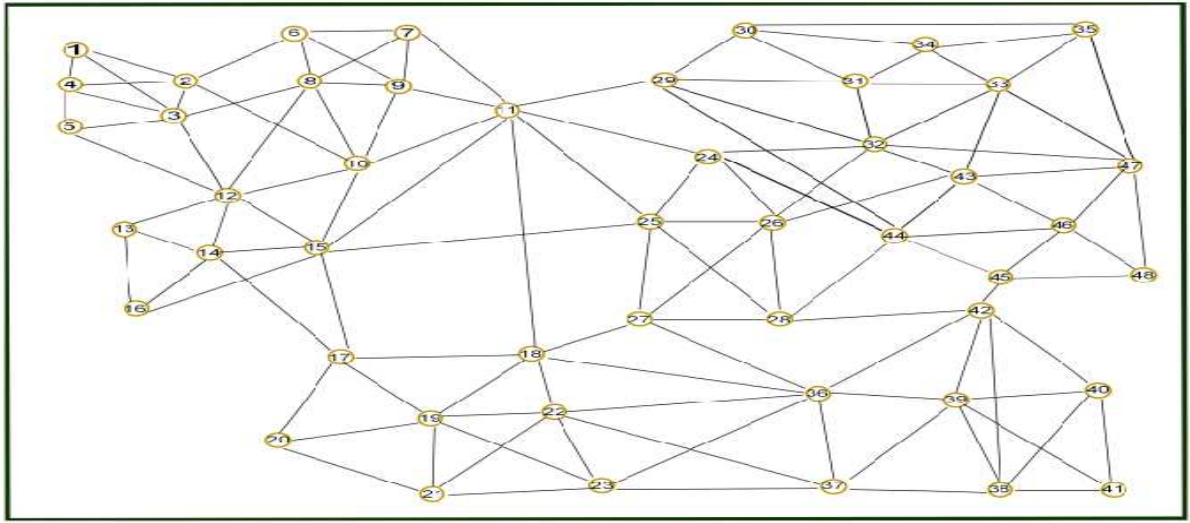}
\caption{Network $T$}
\label{t}
\end{figure}

Application of this algorithm provides a network deconstruction as it gradually eliminates
links until  reaching the  network with no links. Reading the resulting forest from leaves to
roots we obtain a reconstruction of our network, a genealogy of the various interrelated
components of the network. The properties of the network revealed by the deconstruction
procedure very much depend, as we shall see, on the choice of ranking among links.

We address the finding of rankings on links by  reducing it
to finding rankings on nodes: assuming  a ranking method on nodes
we propose three ranking methods on links, two of them  obtained by applying geometric
procedures (dual and barycentric constructions) to build a new network from the given one,
so that nodes of the new network encode information on the links of the original network.
The third one formalizes the intuition that a link is sort of a bridge, and thus its importance
is proportional to its functionality and to the importance of the nodes that it joints.
The first two methods use the ranking of nodes by importance, the third method uses
the rankings by indirect dependence and by indirect influence.

The original application of the Girvan-Newman \cite{girvannewman} decontruction algorithm
to clustering uses the ranking on links given by the betweenness degree, i.e. the number of geodesics (length minimizing directed paths)
passing trough a given link. By eliminating links of high betweennes the deconstruction process  uncovers
clusters. Running the deconstruction algorithm with the three links ranking methods proposed below
we obtain new clustering algorithms.  As we are free to choose the ranking on links in the deconstruction algorithm,
we may as well apply the rankings opposite to the rankings mentioned in the previous
paragraph and  eliminate links of the lowest importance,  thus uncovering
core-periphery structures, with the periphery being the nodes that become isolated early on,
and the core being the nodes in resilient connected components.

\section{Dual Double Weighted Networks}\label{dww4}

\begin{figure}[t]
\centering
\includegraphics[width=16cm, height=7cm]{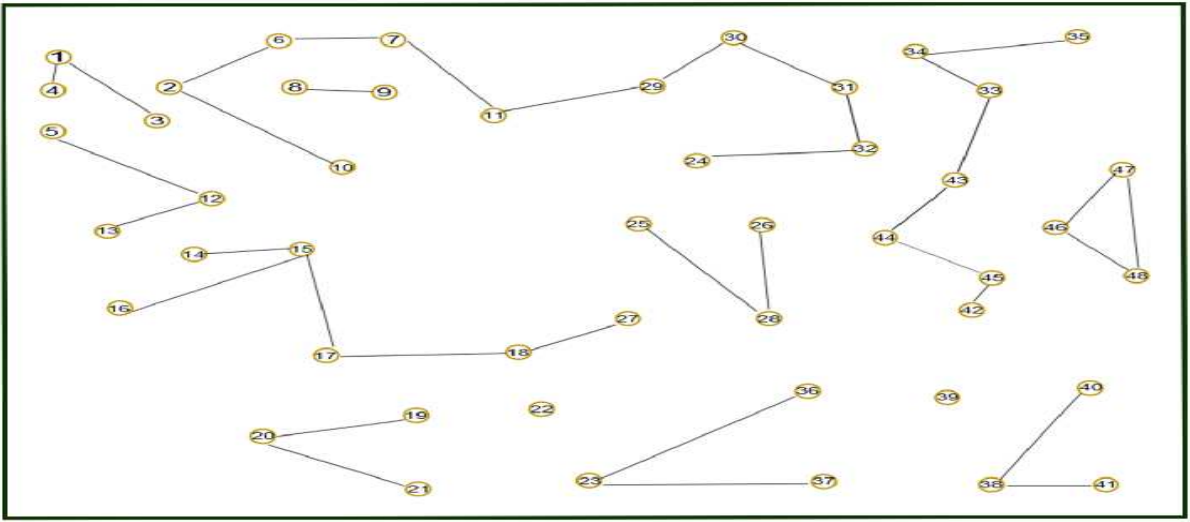}
\caption{Clusters in Network $T$ using Dual Construction}
\label{ct1}
\end{figure}

We introduce the dual  of a double weighted directed network via  the map
$$(\ )^{\star} : \mathrm{wwdigraph} \rightarrow  \mathrm{wwdigraph} $$ sending a network $G$ to its dual
network $G^{\star}.$  Figure \ref{jjj} displays the dual
  network $ Y^{\star} $ to the double weighted network $Y $ from Figure \ref{f1}. Given network
  $ (V, E,s,t, w, f) $ its dual  network
  $  (V^{\star}, E^{\star},s^{\star},t^{\star}, w^{\star}, f^{\star}) $ is such that:

\begin{itemize}
  \item $V^{\star}=  E \ $ and $\ E^{\star}  =  \{(e,v,h) \in E\times V \times E \  |  \ te=v=sh\}.\  $
   For $ e \in  V^{\star} \ $ set $ \ f^{\star}(e)  =  w(e).$
  \item For $ (e,v,h) \in E^{\star} \ $ set  $ \ (s^{\star},t^{\star})(e,v,h) =  (e,h)\ $
  and $\ w^{\star}(e,v,h)  = \frac{f(v)}{\mathrm{out}(v)\mathrm{in}(v)},$
where $\mathrm{out}(v)$ and  $\mathrm{in}(v)$ are the out-degree and in-degree of node $v$.

\end{itemize}

The dual construction applied to networks without multiple links may be identified with the map
$ (\ )^{\star}:  \mathrm{M}_n(\mathbb{R})\times \mathbb{R}^n   \rightarrow
 \mathrm{M}_{n^2}(\mathbb{R})\times \mathbb{R}^{n^2}, $ where
 $\mathrm{M}_{n^2}(\mathbb{R}) $ is the space of maps $[n]^2\times [n]^2
  \rightarrow  \mathbb{R}.$   The map $(\ )^{\star}$ sends  $ (D,f) $
  to the pair $ (D^{\star},f^{\star}) $ given by $ f^{\star}_{(i,j)} =  D_{ij}; \ $
$ D^{\star}_{(i,j)(l,k)}=0$ if either $  j\neq l,  $  or $  D_{ij}=0,   $ or $ D_{lk}=0 ;\ $
and $ D^{\star}_{(i,j)(j,k)}=  \frac{f_j}{\mathrm{out}(j)\mathrm{in}(j)}
 $ if $  D_{ij}\neq 0 \  \mbox{and} \ D_{lk}\neq 0. $

\begin{figure}[t]
\centering
\includegraphics[width=16cm, height=7cm]{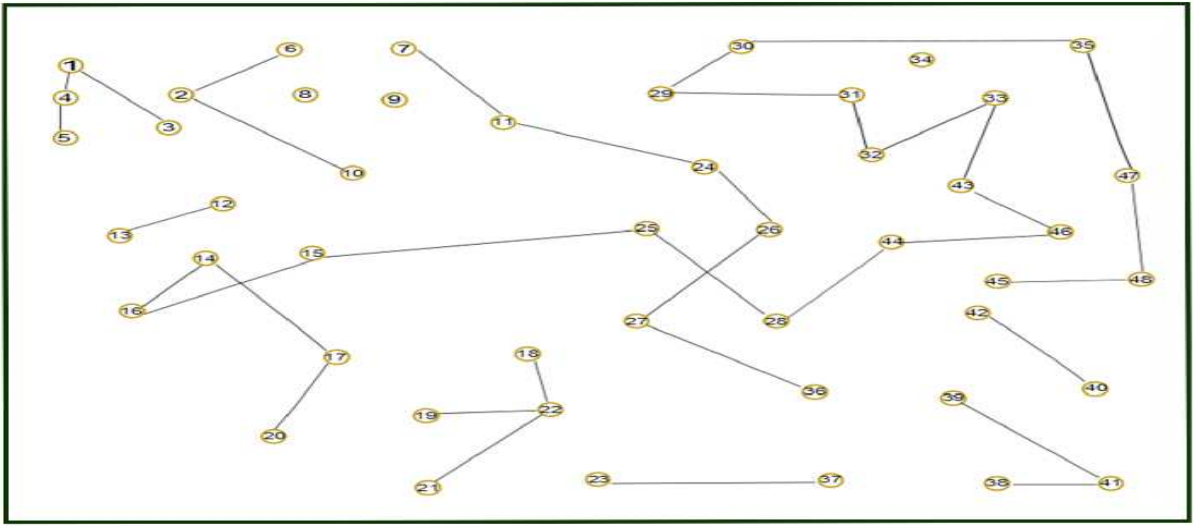}
\caption{Clusters in Network $T$ using Barycentric Division}
\label{ct2}
\end{figure}

A node-ranking map $ \mathrm{R}_n : \mathrm{M}_{n}(\mathbb{R})\times \mathbb{R}^{n}  \rightarrow
 \mathrm{ranking} [n] $ on weighted networks gives rise to link-ranking map
 $ \mathrm{R}_n^{\star} : \mathrm{M}_{n}(\mathbb{R})\times \mathbb{R}^{n}
  \rightarrow  \mathrm{ranking}[n]^2 $  given by $  \mathrm{R}_n^{\star}  =
    \mathrm{R}_{n^2} \circ (\ )^{\star} .$ Looking at the node-ranking maps $  \mathrm{E}, \mathrm{F},\mathrm{I}:
     \mathrm{M}_{n}(\mathbb{R})\times \mathbb{R}^{n}   \rightarrow  \mathrm{ranking} [n]  $
     introduced in Section \ref{pwpwwny}, we obtain the corresponding link-ranking maps
     by indirect dependence, influence, and importance
$ \mathrm{E}^{\star}, \mathrm{F}^{\star}, \mathrm{I}^{\star} :
\mathrm{M}_{n}(\mathbb{R})\times \mathbb{R}^{n}  \rightarrow   \mathrm{ranking}[n]^2. $

\begin{thm}\label{t1}
{\em The PWP matrix of indirect influences on the dual double weighted network  $  (V^{\star}, E^{\star},s^{\star},t^{\star}, w^{\star}, f^{\star}) $
is given
for $e,f \in E$ by
$$ T_{ef}^{\star}  =  \frac{1}{e^{\lambda}-1}\sum_{k=1}^{\infty}\Big(\underset{\underset{te_i=se_{i+1}}{e=e_k, \ldots , e_0=f}}{\sum} \ \frac{f(te_{k-1})w(e_{k-1}) \cdots f(te_0)w(e_0)}{\mathrm{out}(te_{k-1}) \mathrm{in}(te_{k-1})\cdots \mathrm{out}(te_{0})\mathrm{in}(te_{0})}  \Big)\frac{\lambda^k}{k!}.$$
or  in matrix notation
$$ T_{(m,l)(i,j)}^{\star}  =
\frac{1}{e^{\lambda}-1}\sum_{k=1}^{\infty}\Big(\underset{l=i_k, \ldots , i_1=i}
{\sum} \ \frac{f_{i_k}D_{i_ki_{k-1}} \cdots f_{i_1}D_{i_1j}}{\mathrm{out}(i_k)
\mathrm{in}(i_k)\cdots \mathrm{out}(i_1)\mathrm{in}(i_1)}  \Big)\frac{\lambda^k}{k!}.$$
}
\end{thm}

Using the dual network  $ Y^{\star} $ one gets the following ranking on the links of  network $Y:$

\begin{center}
\begin{tabular}{|c|c|}
  \hline
  \ \ & Ranking the links of network $\ Y \ $ dual method \\  \hline
  \ Influence \  & \ $56\ > \ 24 \ >  \ 64\  > \ 31 \ > \ 13 \ >\ 43 >\ 41 >\ 63 \ >\ 15 \ >\ 35$\ \\  \hline
  \ Dependence \ & \ $31,35\ > \ 56 \ > \ 13,15 \ > \ 64,63\ > \ 43,41 \ > \ 24$\ \\  \hline
  \ Importance \ & \ $56\ > \ 31 \ > \ 13 \ > \ 35 \ > \ 15 \ > \ 64\ > \ 63 \ > \ 43 \  > \ 41\  > \ 24$\ \\  \hline

\end{tabular}
\end{center}

We are ready to apply the deconstruction method from Section \ref{dww}, regarded
as a clustering method, ranking links of $ Y $ by importance using the dual method.
Figure \ref{cdb} displays on the left the various stages as we deconstruct network
$ Y  $ until the bare network is reached. Note that nodes are separated into various
components only at the very last step where  all remaining links have the same importance,
as no concatenation of links is even possible.  Thus there is only one cluster  encompassing
all nodes. Figure \ref{cpdb} shows the core-periphery finding process for network
$ Y  $ considering the rank of links by importance and using the dual construction, revealing
a core consisting of four rings: the core $\{1,3\}$, second ring $\{5,6 \},$
third ring $\{4\}$ and the periphery $\{2\}. $

\begin{figure}[t]
\centering
\includegraphics[width=16cm, height=7cm]{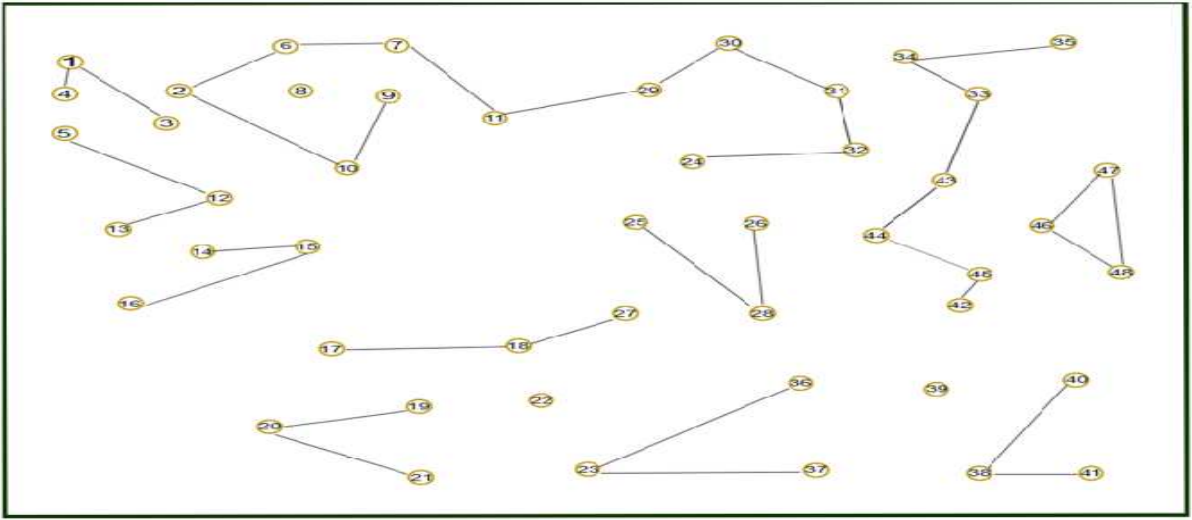}
\caption{Clusters in Network $T$ using Bridge Approach}
\label{ct3}
\end{figure}

\section{Barycentric Division of Double Weighted Networks}\label{dwwn5}

The barycentric division of double weighted networks is a construction that places nodes and links of
networks on the same footing. Figure \ref{jjj} shows the  barycentric division $Y^{\circ}$ of network $Y$.
This construction allows to compare the importance of nodes with the importance of links, thus providing
a precise formulation of the question of whether a network is dominated
by actors or  by relations. The barycentric division map
$$  (\ )^{\circ}: \mathrm{wwdigraph}  \rightarrow  \mathrm{wwdigraph}$$ turns nodes and links
of $(V, E,s,t, w, f)$ into nodes of the network $(V^{\circ}, E^{\circ},s^{\circ},t^{\circ}, w^{\circ}, f^{\circ}) $ defined as follows:
\begin{itemize}
  \item $V^{\circ} =  V \sqcup E \  $ and  $\ E^{\circ}  =  \{(v,e) \in V \times E \ | \  v=se\}  \sqcup  \{(e,v) \in E\times V \  |  \ te=v\}.$
  \item For $(v,e) \in E^{\circ} \ $ set  $\ (s^{\circ},t^{\circ})(v,e)=  (v,e) \in  V^{\circ}\times V^{\circ} .$
  \item For $(e,v) \in E^{\circ} \ $ set  $\ (s^{\circ},t^{\circ})(e,v) =  (e,v)
  \in  V^{\circ} \times V^{\circ} .$
  \item For $(e,v) \ \mbox{and} \ (v,e) \in E^{\circ}\ $ set
  $\  w^{\circ}(v,e)  = w^{\circ}(e,v) =  1.$
  \item For $ v \in  V \subseteq V^{\circ}$  set $ \ f^{\circ}(v) =  f(v).$ \  For
  $ e \in  E \subseteq V^{\circ}\
   $  set $\  f^{\circ}(e) =  w(e).$
\end{itemize}

\begin{figure}[t]
\centering
\includegraphics[width=16cm, height=7cm]{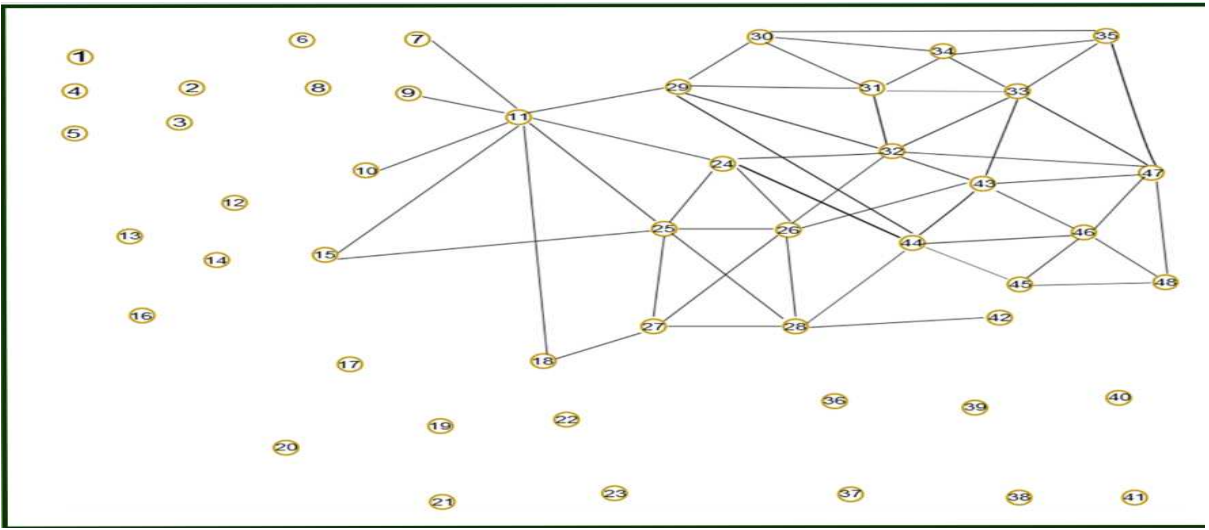}
\caption{Core of Network $T$ using Dual Construction}
\label{z1}
\end{figure}

The barycentric construction applied to networks without multiple links may be identified with the map
$ (\ )^{\circ}:  \mathrm{M}_n(\mathbb{R})\times \mathbb{R}^n   \rightarrow
 \mathrm{M}_{n^2+n}(\mathbb{R})\times \mathbb{R}^{n^2+n}, $
Given  a  node-ranking map
$  \mathrm{R}_n : \mathrm{M}_{n}(\mathbb{R})\times \mathbb{R}^{n}   \rightarrow   \mathrm{ranking}[n] $
on double weighted networks, we construct the link-ranking map
$\mathrm{R}_n^{\circ} : \mathrm{M}_{n}(\mathbb{R})\times \mathbb{R}^{n}   \rightarrow   \mathrm{ranking}[n]^2 $
given by $ R_n^{\circ}   = r\circ R_{n^2+n} \circ (\ )^{\circ} ,$ where
$r: \mathrm{ranking}[n^2+n] = \mathrm{ranking}([n]^2\sqcup [n]) \rightarrow \mathrm{ranking}[n]^2$ is the restriction map.

\begin{thm}\label{t2}
{\em  The PWP matrix of indirect influences on the barycentric division double weighted network $(V^{\circ}, E^{\circ},s^{\circ},t^{\circ}, w^{\circ}, f^{\circ}) $
is given
for $e,f \in E$ by
$$ T_{ef}^{\circ}  =  \frac{1}{e^{\lambda}-1}\sum_{k=1}^{\infty}\Big(\underset{\underset{se_{i+1}=te_i}{e=e_k, \ldots , e_0=f}}{\sum} \ f(te_{k-1})w(e_{k-1}) \cdots f(te_0)w(e_0)  \Big)\frac{\lambda^{2k}}{(2k)!}.$$
or in matrix notation
$$ T_{(m,l)(i,j)}^{\circ}  =  \frac{1}{e^{\lambda}-1}\sum_{k=1}^{\infty}\Big(\underset{l=i_k, \ldots , i_1=i}{\sum} \ f_{i_k}D_{i_ki_{k-1}} \cdots f_{i_1}D_{i_1 j} \Big)\frac{\lambda^{2k}}{(2k)!}.$$
}
\end{thm}

\begin{figure}[t]
\centering
\includegraphics[width=16cm, height=7cm]{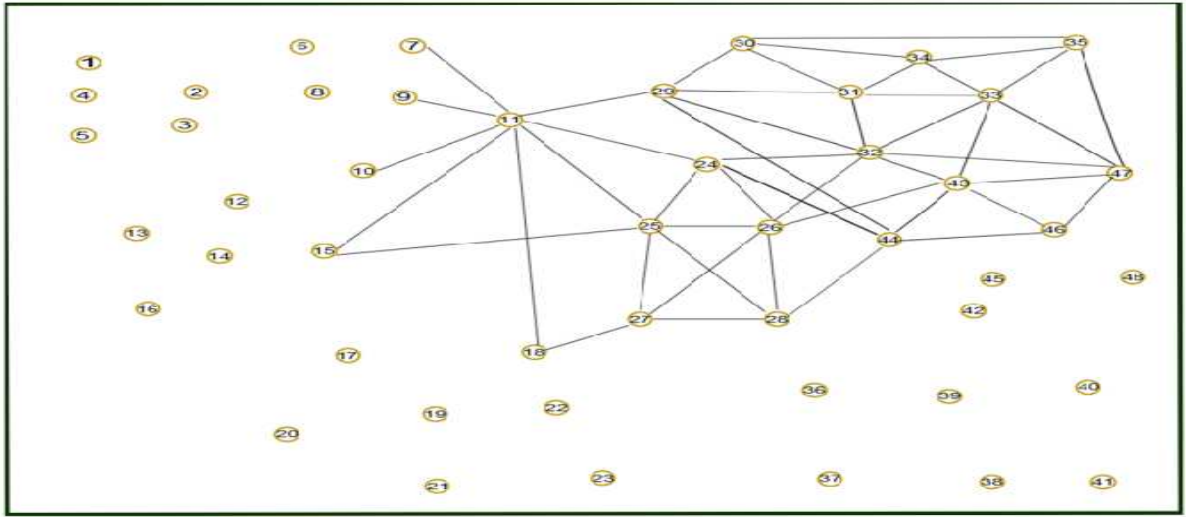}
\caption{Core of Network $T$ using Barycentric Division}
\label{z2}
\end{figure}

Considering the barycentric division network $ Y^{\circ}  $ and the ranking on nodes by indirect
dependence, indirect influence, and importance of its nodes we obtain the ranking on the links  of
network $ Y$
given by:

\begin{center}
\begin{tabular}{|c|c|}
  \hline
   & Ranking on Nodes of Network  $ Y^{\circ}$  \\  \hline
  Influence & $6> 4 > 3,1 > 56 > 24,64 > 63,13,43,31,41 > 5 > 2 > 35,15$ \\  \hline
  Dependence & $3> 5 > 1 > 35,31 > 4 > 56 > 15,13 > 43,41 > 6 > 64,63 > 24 > 2$ \\  \hline
  Importance & $3> 1 > 4 > 5 > 6 > 31 > 56 > 13 > 35 > 43,41 > 15 > 64 > 63 > 24 > 2$ \\  \hline

\end{tabular}
\end{center}

Figure \ref{cdb} displays on the right the clustering process for network $ Y$ based
on the barycentric construction using the ranking by importance on links, again we obtain
just one cluster component $ \{ 1,2,3,4, 5,6\} . $ Figure \ref{cpdb} displays the core-periphery finding
process for  network $Y $ based on the barycentric
construction, yielding the same result as with the dual construction: core $ \{ 1, 3\} $  and subsequent peripheral outer
rings $ \{ 5,6 \}, $    $\{ 4\}, \{ 2\}.$

\begin{figure}[t]
\centering
\includegraphics[width=16cm, height=7cm]{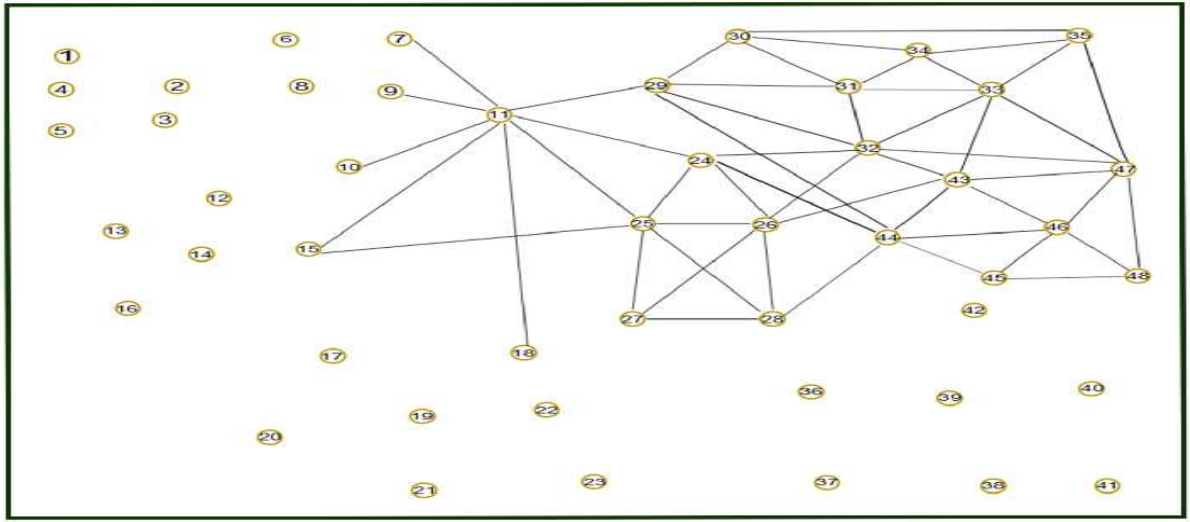}
\caption{Core of Network $T$ using Bridge Approach}
\label{z3}
\end{figure}

\section{Bridge Approach to Link Ranking}\label{dwwn6}

Our third method for ranking links on double weighted network is based on the idea
that the importance of a link is proportional to the importance of the nodes that it connects and to its functionality.
Assume that we have already computed the dependence and influence of nodes as in Section \ref{pwpwwny}.
The dependence  $\mathrm{E}(e), $ influence $ \mathrm{F}(e),  $ and importance $ \mathrm{I}(e) $ of link $e$
in a double weighted network are given, according to the bridge approach, by
$ \mathrm{E}(e) =  \mathrm{E}(se)f(se), \ $ $  \mathrm{F}(e) = w(e)\mathrm{F}(te), \ $
and $  \ \mathrm{I}(e) =  \mathrm{E}(se)f(se)  +  w(e)\mathrm{F}(te).$
The following result is  consequence of Proposition \ref{v}.
\begin{thm}

{\em The importance of link $ e $ in a double weighted directed network is given, according to the bridge approach, by
$$\mathrm{I}(e) =  \frac{f(se)}{e^{\lambda}-1}  \sum_{k=1}^{\infty} \sum_{\underset{te_k=se,
se_{i+1}=te_i}{e_k,...,e_1}} w(e_k)f(se_k)\cdots w(e_1)f(se_1)\frac{\lambda^k}{k!} \ \   + $$
$$\frac{w(e) }{e^{\lambda}-1}\sum_{k=1}^{\infty}
\sum_{\underset{se_{i+1}=te_i,  se_1=te }{e_k,...,e_1}} w(e_k)f(se_k)
\cdots w(e_1)f(se_1) \frac{\lambda^k}{k!}.$$
Equivalently, in matrix notation
$$\mathrm{I}_{ij}=\frac{f_i}{e^{\lambda}-1} \sum_{k=1}^{\infty} \sum_{i_k,...,i_1}
D_{ii_k}f_{i_k}\cdots D_{i_2i_1}f_{i_1}\frac{\lambda^k}{k!} \    + \
\frac{D_{ij}}{e^{\lambda}-1}\sum_{k=1}^{\infty} \sum_{i_k,...,i_1} D_{i_ki_{k-1}}f_{e_{k-1}}\cdots D_{i_1j}f_j\frac{\lambda^k}{k!}.$$

}
\end{thm}

\section{Symmetric and Generic Network Examples}\label{s7}

As a rule one expects the dual, barycentric, and bridge methods to yield different results, as their
explicit formulae given above indicate, nevertheless in some cases they do agree. In this section we consider a highly
symmetric network $S$, shown on the left of Figure \ref{jjj1}, coming with intuitively clear
clustering and core-periphery structures. Network $S$, with $12$ nodes  and $15$ links,  consists of three directed
$4$-cycles connected through $3$ nodes forming an additional directed cycle. The barycentric
network $ S^{\circ},$  with $25 $ nodes and $30$ links,  and the dual network $ S^{\ast},$ with
$15$ nodes and $21$ links, are shown in the center and right hand
side of Figure  \ref{jjj1}. Although $ S^{\circ}$ and $S^{\star}$ are different networks,  our three methods yield the
same clustering and core-component structure, and indeed
the outputs are what one may naively expect: the clusters are the  directed cycles
$\{1,2,3,4 \},\ \{5,6,7,8\},\ \{9,10,11,12 \}$, the core are the nodes \{4,6,9\} connecting these cycles,
with a second layer formed by the nodes  adjacent to the core $\{1,3,5,7,10,12 \}$, and
the periphery being the nodes $\{2,8,11 \}$ attached to the second layer.

Finally, we test our methods on a more sophisticated network $T$ with $48$ nodes and $242$ links  shown in
Figure \ref{t}, which we borrowed from \cite{s}.
Applying our three clustering methods to $T$ until the obtained clusters are trees or cycles, we obtain
the networks displayed in Figures \ref{ct1}, \ref{ct2}, \ref{ct3}, which although not identical are actually
pretty similar. Note however that the barycentric method yields a pretty large cluster with $19$ nodes.
The number of steps required to reach such clusterings  with our three methods are also quite similar.
The core of network $T$ according to our three methods are shown in Figures \ref{z1}, \ref{z2},
\ref{z3}, respectively. Again the outputs are pretty consistent, and were obtained in roughly the same number of steps.

\section{Conclusion}

We have shown that the problems of hierarchization, clustering, and core-periphery finding
are intimately related. Indeed any  hierarchization method, together with a suitable choice of
network constructions, leads to clustering and core-periphery finding methods.
We considered three construction, namely, the dual, barycentric subdivision, and bridge constructions
on double weighted networks. Applying this philosophy together with the PWP method for ranking nodes
we obtain new clustering and  core-periphery methods, which we computed in three toy models. We also
applied the PWP map to obtain a one-parameter deformation
of the modularity function. Our methods can be readily be modified to use other definitions
for the matrix of indirect influences,  instead of the PWP map.

\

\

\noindent ragadiaz@gmail.com\\
\noindent Universidad Nacional de Colombia - Sede Medell\'in, Facultad de Ciencias,\\
Escuela de Matem\'aticas, Medell\'in, Colombia\\

\noindent jorgerev90@gmail.com, \ \ angelikius90@gmail.com  \\
\noindent Departamento de Matem\'aticas, Universidad Sergio Arboleda, \ Bogot\'a, \ Colombia\\

\end{document}